\documentclass[submission,copyright,creativecommons]{eptcs}
\providecommand{\event}{GandALF 2016} 
\usepackage{breakurl}             

\title{Weighted Linear Dynamic Logic}
\author{Manfred Droste
\institute{Institut f\"{u}r Informatik \\ Universit\"{a}t Leipzig \\D-04109 Leipzig, Germany}
\email{droste@informatik.uni-leipzig.de}
\and
George Rahonis
\institute{Department of Mathematics \\Aristotle University of Thessaloniki\\54124 Thessaloniki, Greece}
\email{grahonis@math.auth.gr}
}
\def\titlerunning{Weighted Linear Dynamic Logic}
\def\authorrunning{M. Droste \& G. Rahonis}

\usepackage{amssymb}
\usepackage{amsfonts}
\usepackage{amsmath}
\usepackage{amsbsy}

\newtheorem{theorem}{Theorem}

\newtheorem{definition}[theorem]{Definition}
\newtheorem{example}[theorem]{Example}
\newtheorem{lemma}[theorem]{Lemma}
\newtheorem{corollary}[theorem]{Corollary}
\newtheorem{remark}[theorem]{Remark}

\newtheorem{proposition}[theorem]{Proposition}

\newenvironment{proof}{\textbf{Proof.}}

\begin{document}

\maketitle

\begin{abstract}
We introduce a weighted linear dynamic logic (weighted \emph{LDL} for short) and
show the expressive equivalence of its formulas to weighted rational
expressions. This adds a new characterization for recognizable series to the fundamental Sch\"{u}tzenberger theorem. Surprisingly, the equivalence does not require any restriction to our weighted \emph{LDL}. Our results hold over arbitrary (resp. totally complete)
semirings for finite (resp. infinite) words. As a consequence, the equivalence problem for weighted \emph{LDL} formulas over fields is decidable in doubly exponential time. In contrast to classical logics, we show that our weighted \emph{LDL} is expressively incomparable to weighted \emph{LTL} for finite words. We determine a fragment of the weighted \emph{LTL} such that series over finite and infinite words definable by \emph{LTL} formulas in this fragment are definable also by weighted \emph{LDL} formulas. 
 \end{abstract}

\section{Introduction}

Linear Temporal Logic (\emph{LTL} for short) is widely used in several areas
of Computer Science like, for instance in model
checking where it plays the role of a specification language
\cite{Ba:Pr,Gi:Li}, and in artificial intelligence \cite{Gi:Li}. Nevertheless, \emph{LTL} formulas are expressively weaker
than finite automata, namely the class of \emph{LTL}-definable languages
coincides with the class of First-Order (\emph{FO} for short) logic definable
languages (cf. \cite{Di:Fi} for an excellent survey on the topic). Therefore,
it was greatly desirable, especially for applications, to
have a logic which combines the complexity properties of reasoning on
\emph{LTL} and the expressive power of finite automata. This was recently
achieved in \cite{Gi:Li}, where the authors introduced a \emph{Linear Dynamic
Logic }(\emph{LDL\ }for short) which is a combination of \emph{Propositional
Dynamic Logic} (cf. \cite{Ha:Dy}) and \emph{LTL}. The satisfiability, validity, and logical
implication of \emph{LDL }formulas interpreted over finite words were proved
to be PSPACE-complete \cite{Gi:Li,Gi:Sy}, as for \emph{LTL}. This was obtained by a translation of \emph{LDL }formulas
to finite automata. Similar
results were stated for \emph{LDL }formulas interpreted over infinite words in
\cite{Va:Th}.

In the weighted setup, a B\"{u}chi type theorem stating the coincidence of recognizable series with the ones
defined in a fragment of a weighted Monadic Second-Order (\emph{MSO }for
short) logic over semirings, was firstly proved in \cite{Dr:We} (cf. also
\cite{Dr:Wh}). Then, weighted
\emph{MSO\ }logics have been investigated for several objects, including trees,
pictures, nested words, graphs, and timed words. The weight structure of
the semiring has been also replaced by more general ones incorporating average or discounting of weights. Most of the results work for finite as well as infinite objects. A weighted version
of \emph{LTL }over De Morgan algebras was firstly introduced in \cite{Ku:La}.
In \cite{Dr:Mu} the authors proved several characterizations of \emph{LTL}%
-definable and \emph{LTL}-$\omega$-definable series over arbitrary bounded
lattices. Recently, a weighted \emph{LTL} with averaging modalities was studied in \cite{Bo:Av}, and a weighted \emph{LTL} over idempotent and zero-divisor free semirings satisfying
completeness axioms was investigated in \cite{Ma:Ph,Ma:Se}. In \cite{Al:Fo,Al:Di} the authors considered a discounted \emph{LTL} with values in $[0,1]$ and in \cite{Ma:Ph,Ma:We} in the max-plus semiring.

It is the goal of this paper to introduce and investigate a \emph{weighted LDL} over
arbitrary semirings. Our work is motivated as follows. In recent applications like verification of systems \cite{Ch:Qu} and artificial intelligence (cf. for instance \cite{La:Th}), classical automata have been replaced by quantitative ones. Therefore it is highly desirable to have a quantitative logic which is expressively equivalent to weighted automata. However, the class of series which are definable by all weighted \emph{MSO} logic sentences exceeds that of recognizable series. Furthermore, the weighted \emph{FO} logic over finite words is, in general, expressively incomparable to weighted finite automata \cite{Dr:We}, and this is shown here also for the weighted \emph{LTL}.  Therefore, in view of the results of \cite{Gi:Li,Gi:Sy} for \emph{LDL}, we investigate weighted \emph{LDL}. We show that our weighted \emph{LDL} is expressively equivalent to weighted finite automata over semirings. Surprisingly, there is no need to consider, as for the weighted \emph{MSO} logic, any fragment of our logic to achieve the aforementioned equivalence. Our results hold for finite and infinite words and this shows the robustness of our theory and in turn the robustness of the \emph{LDL} of \cite{Gi:Li,Gi:Sy,Va:Th}.  Our main results are as follows. 
\begin{itemize}
\item The class of \emph{LDL}-definable series coincides with the class of generalized rational series over arbitrary semirings.
\item The class of \emph{LDL}-definable series coincides with the class of recognizable series over commutative semirings. This extends the fundamental Sch\"{u}tzenberger theorem, for commutative semirings, with a logic directed characterization.
\item The equivalence problem for weighted \emph{LDL} formulas is decidable in doubly exponential time for a large class of weight structures including computable fields, as the realizability problem for \emph{LDL} \cite{Gi:Sy}.    
\item The class of \emph{LDL}-$\omega$-definable series coincides with the class of generalized $\omega$-rational series over totally complete semirings.
\item The class of \emph{LDL}-$\omega$-definable series coincides with the class of $\omega$-recognizable series over totally commutative complete semirings. 
\end{itemize}
Our weighted \emph{LDL} consists of the classical, unweighted \emph{LDL}
of \cite{Gi:Li} with the same interpretation and a
copy of it which is interpreted quantitatively.
Therefore, practitioners can use the classical \emph{LDL} part
as they are used to, and the copy of it in the same way
to compute quantitative interpretation. A similar approach
was followed for weighted \emph{MSO} logic recently in \cite{Ga:Un}. While the translation of the restricted weighted \emph{MSO} logic formulas of \cite{Dr:Wh}
to weighted automata as for \emph{MSO} is non-elementary,
the translation of the present weighted \emph{LDL} into weighted
automata can be done in doubly exponential time, as for \emph{LDL}.
We prove that our weighted \emph{LDL} interpreted over finite words, is in general expressively incomparable to weighted \emph{LTL} of \cite{Ma:Ph,Ma:Se}. We define a fragment of that weighted \emph{LTL} and prove that series over finite and infinite words definable by weighted \emph{LTL} formulas in this fragment are definable as well by weighted \emph{LDL} formulas. Furthermore, our weighted \emph{LDL} is expressively equivalent to weighted conjunction-free $\mu$-calculus \cite{Me:We} for a particular class of semirings.

\section{Semirings and rational operations}

Let $A$ be an alphabet, i.e., a finite nonempty set. As usually, we denote by
$A^{\ast}$ (resp. $A^{\omega}$) the set of all finite (resp. infinite) words
over $A$ and $A^{+}=A^{\ast}\setminus\{\varepsilon\}$, where $\varepsilon$ is
the empty word. We write a finite (resp. infinite) word often as $w=w(0)\ldots w(n-1)$ (resp. $w=w(0)w(1)\ldots$) where $w(i) \in A$ for every
$i\geq0$. For every finite (resp. infinite) word $w=w(0)\ldots w(n-1)$ (resp.
$w=w(0)w(1)\ldots$) and every $0\leq i\leq n-1$ (resp. $i\geq0$) we denote by
$w_{\geq i}$ the suffix $w(i)\ldots w(n-1)$ (resp. $w(i)w(i+1)\ldots$) of $w$.
In the sequel, we use the letter $a$ with indices to denote the elements of an
alphabet $A$.

A \emph{semiring\/ }$(K,+,\cdot,0,1)$ is denoted
simply by $K$ if the operations and the constant elements are understood. If
no confusion is caused, we shall denote the operation $\cdot$ simply by
concatenation. The result of the empty product as usual equals to $1$.

$\quad$ \emph{Throughout the paper }$A$\emph{ will denote an alphabet} and $K$\emph{ a
semiring.}

A \emph{formal series} (or simply \emph{series}) \emph{over}
$A^*$ \emph{and} $K$ is a mapping $s:A^*\rightarrow K$. We denote by $K\left\langle \left\langle A^*\right\rangle
\right\rangle $ the class of all series over $A^*$ and $K$. The \emph{constant series}
$\widetilde{k}$ ($k\in K$) is defined, for every $w\in A^*$,$\ $by
$\widetilde{k}(w)=k$. The \emph{characteristic series }$1_{L}$ \emph{of a language
}$L\subseteq A^*$ is given by $1_{L}(w)=1$ if $w\in L$ and $1_{L}(w)=0$
otherwise. If $L=\{w\}$ is a singleton, then we write $w$
in place of $1_{\{w\}}$.  
Let $s,r\in K\left\langle \left\langle A^*\right\rangle \right\rangle $ and
$k\in K$. The \emph{sum} $s+r$, the \emph{products with scalars} $ks$ and $sk$
as well as the\emph{ Hadamard product} $s\odot r$\ are defined elementwise by
$\ (s+r)(w)=s(w)+r(w),$ $(ks)(w)=ks(w),$ \ \ $(sk)(w)=s(w)k,$ $\ (s\odot
r)(w)=s(w)r(w)$ \ for every $w\in A^*$. Trivially, the
structure $\left(  K\left\langle \left\langle A^*\right\rangle \right\rangle
,+,\odot,\widetilde{0},\widetilde{1}\right)  $ is a semiring. The \emph{Cauchy product }$s\cdot
r\in K\left\langle \left\langle A^{\ast}\right\rangle \right\rangle $ is determined by $(s\cdot r)(w)=\sum\nolimits_{w=uv}s(u)r(v)$ for every $w\in A^{\ast}$. The $n$\emph{th-iteration }$s^{n}\in K\left\langle
\left\langle A^{\ast}\right\rangle \right\rangle $ ($n\geq0$) is
defined inductively by $s^{0}=\varepsilon$ and $s^{n+1}=s\cdot s^{n}$ for
every $n\geq 0$. The series $s$ is called \emph{proper} if $s(\varepsilon)=0$.
If $s$ is proper, then for every $w\in A^{\ast}$ and $n>\left\vert
w\right\vert $ we have $s^{n}(w)=0$. The \emph{iteration }$s^{+}\in
K\left\langle \left\langle A^{\ast}\right\rangle \right\rangle $ \emph{of a
proper series }$s$ is defined by $s^{+}=\sum\nolimits_{n>0}s^{n}$. 

The
class of \emph{weighted rational expressions over }$A$ \emph{and} $K$ \cite{Dr:Han} is given
by the grammar $E::=ka\mid E+E\mid E\cdot E\mid E^{+}$ 
where $k\in K$ and $a\in A\cup\{\varepsilon\}$. We denote by $RE(K,A)$ the
class of all such weighted rational expressions over $A$ and $K$. For the
relationship with weighted logics, we will need to consider the Hadamard
product as a rational operation. Therefore, we introduce the class of
\emph{generalized weighted rational expressions over }$A$ \emph{and} $K$ which
is given by the grammar $E::=ka\mid E+E\mid E\cdot E\mid E^{+}\mid E\odot E $, where $k\in K$ and $a\in A\cup\{\varepsilon\}$. We shall denote by $GRE(K,A)$
the class of generalized weighted rational expressions over $A$ and $K$. The
\emph{semantics} of a (generalized) weighted rational expression $E$ is a
series $\left\Vert E\right\Vert \in K\left\langle \left\langle A^{\ast
}\right\rangle \right\rangle $ which is defined inductively by $\left\Vert ka\right\Vert =ka,  \ \ \left\Vert E+E^{\prime}\right\Vert =\left\Vert E\right\Vert
+\left\Vert E^{\prime}\right\Vert ,  \ \ \left\Vert E\cdot E^{\prime}\right\Vert =\left\Vert E\right\Vert
\cdot\left\Vert E^{\prime}\right\Vert , \ \  \left\Vert E^{+}\right\Vert =\left\Vert E\right\Vert ^{+}$ (if $\left\Vert E\right\Vert $ is proper; otherwise undefined),  \ \  $ \left\Vert E\odot E^{\prime}\right\Vert =\left\Vert E\right\Vert
 \odot\left\Vert E^{\prime}\right\Vert$. 
A series $s\in K\left\langle \left\langle A^{\ast}\right\rangle \right\rangle
$ is called \emph{rational }(resp. \emph{g-rational}) if there is a weighted
(resp. generalized weighted) rational expression $E$ such that $s=\left\Vert
E\right\Vert $. The following result is the fundamental Sch\"{u}tzenberger theorem stating the coincidence of rational and recognizable series, i.e., series accepted by weighted automata. For the theory on weighted automata we refer the reader to \cite{Es:Ha,Sa:Han,Dr:Au}.

\begin{theorem}
\label{Kleene}\cite{Sc:On,Es:Ha,Sa:Han}
Let $K$ be a semiring and $A$\ an alphabet.
Then a series $s\in K\left\langle \left\langle A^{\ast}\right\rangle
\right\rangle $ is rational iff it is recognizable.
\end{theorem}

It is well-known (cf. \cite{Sc:Th,Be:Ra,Dr:Han}) that if the semiring $K$ is commutative, then the class of recognizable series over $A$ and $K$ is closed under Hadamard product.
Consequently, if $K$ is commutative, then a series $s\in K\left\langle \left\langle A^{\ast}\right\rangle
\right\rangle $ is g-rational iff it is recognizable.

\section{Weighted linear dynamic logic on finite words}

In this section, we introduce the weighted linear dynamic logic (weighted \emph{LDL} for short). Our main
result states the coincidence of the classes of g-rational series and series definable by weighted
\emph{LDL} formulas. First, we recall the \emph{LDL} from
\cite{Gi:Li}. For the definition of our weighted \emph{LDL} below, we need to
modify the notations used for the semantics of \emph{LDL} formulas
in \cite{Gi:Li}.
For every letter $a\in A$ we consider an atomic proposition $p_{a}$ and we let
$P=\{p_{a}\mid a\in A\}$. For every $p\in P$ we identify $\lnot\lnot p$ with
$p.$

\begin{definition}
\label{LDL_fin_synt}The syntax of LDL formulas $\psi$ \emph{over} $A$ is given by the grammar
\begin{align*}
\psi &  ::=true\mid p_{a}\mid\lnot\psi\mid\psi\wedge\psi\mid\left\langle
\theta\right\rangle \psi\\
\theta &  ::=\phi\mid\psi?\mid\theta+\theta\mid\theta;\theta\mid\theta^{+}%
\end{align*}
where $p_{a}\in P$ and $\phi$ denotes a propositional formula over the atomic
propositions in $P$.
\end{definition}

Next, for every \emph{LDL }formula $\psi$ and $w\in A^{\ast}$ we define the
satisfaction relation $w\models\psi$, inductively on the structure of $\psi$,
as follows:

\begin{itemize}
\item[-] $w\models true,$

\item[-] $w\models p_{a}$ \ iff $\ w(0)=a,$

\item[-] $w\models\lnot\psi$ \ iff $\ w\not \models \psi,$

\item[-] $w\models\psi_{1}\wedge\psi_{2}$ \ iff $\ w\models\psi_{1}$ and
$w\models\psi_{2},$

\item[-] $w\models\left\langle \phi\right\rangle \psi$ \ iff $\ w\models\phi$
and $w_{\geq1}\models\psi,$

\item[-] $w\models\left\langle \psi_{1}?\right\rangle \psi_{2}$ \ iff
$\ w\models\psi_{1}$ and $w\models\psi_{2},$

\item[-] $w\models\left\langle \theta_{1}+\theta_{2}\right\rangle \psi$ \ iff
$\ w\models\left\langle \theta_{1}\right\rangle \psi$ or $w\models\left\langle
\theta_{2}\right\rangle \psi,$

\item[-] $w\models\left\langle \theta_{1};\theta_{2}\right\rangle \psi$ \ iff
$\ w=uv$, $u\models\left\langle \theta_{1}\right\rangle true$, and
$v\models\left\langle \theta_{2}\right\rangle \psi,$

\item[-] $w\models\left\langle \theta^{+}\right\rangle \psi$ \ iff \ there
exists $n$ with $1\leq n\leq\left\vert w\right\vert $ such that $w\models
\left\langle \theta^{n}\right\rangle \psi,$
\end{itemize}

\noindent where $\theta^{n},$ $n\geq1$ is defined inductively by $\theta
^{1}=\theta$ and $\theta^{n}=\theta^{n-1};\theta$ for $n>1$.

We let $false=\lnot true$. For an \emph{LDL
}formula $\psi$, we let $L(\psi)=\{w\in A^{\ast}\mid w\models\psi\}$, the language defined by $\psi$. A
language $L\subseteq A^{\ast}$ is called \emph{LDL-definable }if there is an
\emph{LDL} formula $\psi$ such that $L=L(\psi)$. 

\begin{theorem}
\label{LDL_RE}\cite{Gi:Li} A language $L\subseteq A^{\ast}$ is LDL-definable
iff $L$ is rational.
\end{theorem}

\begin{definition}
\label{wLDL_fin_synt}The syntax of formulas $\varphi$ of the \emph{weighted} LDL \emph{over} $A$ \emph{and} $K$
is given by the grammar
\begin{align*}
\varphi &  ::=k\mid\psi\mid\varphi\oplus\varphi\mid\varphi\otimes\varphi
\mid\left\langle \rho\right\rangle \varphi\\
\rho &  ::=\phi\mid\varphi?\mid\rho\oplus\rho\mid\rho\cdot\rho\mid\rho
^{\oplus}%
\end{align*}
where $k\in K$, $\phi$ denotes a propositional formula over
the atomic propositions in $P$, and $\psi$ denotes an LDL formula as in Definition \ref{LDL_fin_synt}.
\end{definition}

We denote by $LDL(K,A)$ the set of all weighted \emph{LDL} formulas $\varphi$
over $A$ and $K$. We represent the semantics $\left\Vert \varphi\right\Vert $
of formulas $\varphi\in LDL(K,A)$ as series in $K\left\langle \left\langle
A^{\ast}\right\rangle \right\rangle $. For the semantics of \emph{LDL}
formulas $\psi$ we use the satisfaction relation as defined above. 

\begin{definition}
\label{def-sem}Let $\varphi\in LDL(K,A)$. The \emph{semantics} of $\varphi$ is
a series $\left\Vert \varphi\right\Vert \in K\left\langle \left\langle
A^{\ast}\right\rangle \right\rangle $. For every $w\in A^{\ast}$ the value
$\left\Vert \varphi\right\Vert (w)$ is defined inductively as follows:

$\begin{array}[c]{ll}
\left\Vert k\right\Vert (w)=k,  &      \left\Vert \varphi_{1}\oplus\varphi_{2}\right\Vert (w)=\left\Vert
\varphi_{1}\right\Vert (w)+\left\Vert \varphi_{2}\right\Vert (w), \\
\left\Vert \psi\right\Vert (w)=\left\{
\begin{array}
[c]{rl}%
1 & \text{if }w\models\psi\\
0 & \text{otherwise}%
\end{array}
\right.  ,         &     \left\Vert \varphi_{1}\otimes\varphi_{2}\right\Vert (w)=\left\Vert
\varphi_{1}\right\Vert (w)\cdot\left\Vert \varphi_{2}\right\Vert (w), \\
\left\Vert \left\langle \phi\right\rangle \varphi\right\Vert
(w)=\left\Vert \phi\right\Vert (w)\cdot\left\Vert \varphi\right\Vert
(w_{\geq1}),  &    \left\Vert \left\langle \varphi_{1}?\right\rangle \varphi
_{2}\right\Vert (w)=\left\Vert \varphi_{1}\right\Vert (w)\cdot\left\Vert
\varphi_{2}\right\Vert (w),\\
\left\Vert \left\langle \rho_{1}\oplus\rho_{2}\right\rangle
\varphi\right\Vert (w)=\left\Vert \left\langle \rho_{1}\right\rangle
\varphi\right\Vert (w)+\left\Vert \left\langle \rho_{2}\right\rangle
\varphi\right\Vert (w),  &    \left\Vert \left\langle \rho^{\oplus}\right\rangle \varphi
\right\Vert (w)=\underset{n \geq 1 }{\sum}\left\Vert
\left\langle \rho^{n}\right\rangle \varphi\right\Vert (w), 
\end{array}$

$\begin{array}[c]{ll}
\left\Vert \left\langle \rho_{1}\cdot\rho_{2}\right\rangle
\varphi\right\Vert (w)=\underset{w=uv}{\sum}\left(  \left\Vert \left\langle
\rho_{1}\right\rangle true\right\Vert (u)\cdot\left\Vert \left\langle \rho
_{2}\right\rangle \varphi\right\Vert (v)\right),
\end{array}$

\noindent where for the definition of $\left\Vert \left\langle \rho^{\oplus}\right\rangle \varphi
\right\Vert (w)$ we assume that $\left\Vert \left \langle \rho \right\rangle true \right\Vert$ is proper, and $\rho^{n},$ $n\geq1$ is defined inductively by $\rho^{1}=\rho$
and $\rho^{n}=\rho^{n-1} \cdot\rho$ for $n>1$.
\end{definition}

A series $s\in K\left\langle \left\langle A^{\ast}\right\rangle \right\rangle
$ is called \emph{LDL-definable }if there is a formula $\varphi\in LDL(K,A)$
such that $s=\left\Vert \varphi\right\Vert $. 
For $K=\mathbb{B}$ (the Boolean semiring) and any $L \subseteq A^*$, clearly $L$ is \emph{LDL}-definable iff $1_L \in \mathbb{B} \left\langle \left\langle A^{\ast}\right\rangle \right\rangle$ is \emph{LDL}-definable, and therefore our weighted \emph{LDL} generalizes \emph{LDL}.

\begin{example}\label{ex1}
We consider the semiring $(\mathbb{N},+, \cdot, 0, 1)$ of natural numbers, $a \in A$, $k \in \mathbb{N} \setminus \{0\}$, and  the weighted LDL formula 
$$\varphi  = \left\langle  \left( \left( \left\langle \left(k\otimes p_{a}\right)  ?\right\rangle Last \right)? \cdot \left( \left\langle \left(k\otimes p_{a}\right)  ?\right\rangle Last \right)? \right) ^{\oplus} \right\rangle true \oplus \bigwedge\nolimits_{a'\in A}\lnot p_{a'} ,$$
where $Last$ denotes the LDL formula $Last::=\left\langle true\right\rangle \bigwedge\nolimits_{a'\in A}\lnot p_{a'}$.  
For every $w=a_{0}\ldots a_{n-1}\in A^{\ast}$ and $0\leq i\leq
n-1$ we get  
\[
w_{\geq i}\models Last\text{ \ iff \ }w_{\geq i+1}\not \models p_{a'} \ \text{ for every }a'\in A\text{ \ iff
\ }i=n-1,
\]
and we can easily see that $ \left\Vert \varphi \right\Vert(w)= k^{2n}$ whenever $w=a^{2n}$ for some $n\geq 0$, and $ \left\Vert \varphi \right\Vert(w)= 0$ otherwise. Furthermore, the series $ \left\Vert \varphi \right\Vert$ is not definable by any weighted FO logic sentence (cf. \cite{Dr:We}) or weighted LTL formula (cf. Section 5). Indeed, let us assume that there is a weighted FO logic sentence (resp. LTL formula) $\varphi '$ such that $ \left\Vert \varphi ' \right\Vert =  \left\Vert \varphi \right\Vert$. Then, by replacing the non zero weights in $\varphi '$ with $true$ we get an FO logic sentence (resp. LTL formula) $\varphi ''$ whose language is $(aa)^*$, which is impossible (cf. \cite{Di:Fi}).   
\end{example}

Next we show that generalized weighted rational expressions can be translated to weighted \emph{LDL} formulas in linear time. 
\begin{theorem}
\label{RE_to_LDL}For every generalized weighted rational expression $E\in
GRE(K,A)$ we can construct, in linear time, a weighted LDL formula
$\varphi_{E}\in LDL(K,A)$ with $\left\Vert \varphi_{E}\right\Vert =\left\Vert
E\right\Vert $.
\end{theorem}
\begin{proof}
[Sketch]
We proceed by induction on the structure of generalized weighted
rational expressions in $GRE(K,A)$. For this, we define for every $E\in
GRE(K,A)$ the weighted \emph{LDL }formula $\varphi_{E}\in LDL(K,A)$ as follows.

\begin{itemize}
\item[-] If $E=k\varepsilon$ with $k\in K$, then $\varphi_{E}=k\otimes
\bigwedge\nolimits_{a\in A}\lnot p_{a}$.

\item[-] If $E=ka$ with $k\in K,a\in A$, then $\varphi_{E}=\left\langle \left(
k\otimes p_{a}\right)  ?\right\rangle Last$.

\item[-] If $E=E_{1}+E_{2}$, then $\varphi_{E}=\varphi_{E_{1}}\oplus
\varphi_{E_{2}}$.

\item[-] If $E=E_{1}\cdot E_{2}$, then $\varphi_{E}=\left\langle \varphi_{E_{1}%
}?\cdot\varphi_{E_{2}}?\right\rangle true$.

\item[-] If $E=E_{1}^{+}$, then $\varphi_{E}=\left\langle \left(  \varphi_{E_{1}%
}?\right)  ^{\oplus}\right\rangle true$.

\item[-] If $E=E_{1}\odot E_{2}$, then $\varphi_{E}=\left\langle \varphi_{E_{1}%
}?\right\rangle \varphi_{E_{2}}$. \hfill $\square$ \medskip
\end{itemize}

\end{proof}

The next theorem shows that also the converse result holds. More precisely, we
show that for every $\varphi\in LDL(K,A)$ we can construct a generalized
weighted rational expression $E_{\varphi}\in GRE(K,A)$ such that $\left\Vert
E_{\varphi}\right\Vert =\left\Vert \varphi\right\Vert $. For this, we first
translate every \emph{LDL} formula into a rational expression using Theorem
\ref{LDL_RE}. The complexity of an inductive translation would be non-elementary since for
every occurrence of a negation symbol we need an exponential complementation
construction. However, one can follow the translation of \cite{Gi:Li,Gi:Sy} with a doubly exponential construction. We shall need the following lemma.

\begin{lemma}
\label{reg_expr_0_1}Let $E$ be a rational expression over $A$\ and $L(E)$ the
language defined by $E$. Then, there is an $E^{\prime}\in RE(K,A)$ such that
$\left\Vert E^{\prime}\right\Vert (w)=1$ if $w\in L(E)$ and $\left\Vert
E^{\prime}\right\Vert (w)=0$ otherwise, for every $w\in A^{\ast}$.
\end{lemma}\smallskip
\begin{proof}
[Sketch] We consider a deterministic automaton for the rational expression $E$ and construct a weighted automaton over $A$ and $K$, with weights $0$ and $1$. \hfill $\square$
\end{proof}\medskip

\begin{theorem}
\label{LDL_to_RE} For every weighted LDL formula $\varphi\in LDL(K,A)$ we can
construct a generalized weighted rational expression $E_{\varphi}\in GRE(K,A)$
such that $\left\Vert E_{\varphi}\right\Vert =\left\Vert \varphi\right\Vert $.
\end{theorem} \smallskip
\begin{proof}
We proceed by induction on the structure of $LDL(K,A)$ formulas
$\varphi$. If $\varphi=\psi$ is an \emph{LDL }formula, then by Theorem
\ref{LDL_RE} it is expressively equivalent to a
rational expression $E_{\psi}$. Then, by Lemma \ref{reg_expr_0_1}, we can
assume that $E_{\psi}$ is a weighted rational expression in $RE(K,A)$, hence in
$GRE(K,A)$, whose semantics gets values $0$ and $1$ and we get $\left\Vert
E_{\psi}\right\Vert =\left\Vert \psi\right\Vert $. Next, assume that
$\varphi=k\in K$. It is straightforward that the generalized weighted rational
expression $E_{\varphi}=k\varepsilon+k\varepsilon\cdot\left(  1A\right)  ^{+}%
$, where $1A=\sum\nolimits_{a\in A}a$, satisfies our claim. If $\varphi
=\varphi_{1}\oplus\varphi_{2}$ or $\varphi=\varphi_{1}\otimes\varphi_{2}$,
then we get our result by the induction hypothesis and the closure of
generalized weighted rational expressions under sum and Hadamard product,
respectively. Now assume that $\varphi=\left\langle \phi\right\rangle
\varphi^{\prime}$. By the induction hypothesis there are $E_{\phi}%
,E_{\varphi^{\prime}}\in GRE(K,A)$ such that $\left\Vert E_{\phi}\right\Vert
=\left\Vert \phi\right\Vert $ and $\left\Vert E_{\varphi^{\prime}}\right\Vert
=\left\Vert \varphi^{\prime}\right\Vert $. We let $E_{\varphi}=E_{\phi}%
\odot(1A\cdot E_{\varphi^{\prime}})$ and we get
\begin{align*}
\left\Vert E_{\varphi}\right\Vert (w)  &  =\left\Vert E_{\phi}\right\Vert
(w)\cdot\left\Vert 1A\cdot E_{\varphi^{\prime}}\right\Vert (w)\\
&  =\left\Vert E_{\phi}\right\Vert (w)\cdot\left\Vert 1A\right\Vert
(w(0))\cdot\left\Vert E_{\varphi^{\prime}}\right\Vert (w_{\geq1})\\
&  =\left\Vert E_{\phi}\right\Vert (w)\cdot\left\Vert E_{\varphi^{\prime}%
}\right\Vert (w_{\geq1})\\
&  =\left\Vert \phi\right\Vert (w)\cdot\left\Vert \varphi^{\prime}\right\Vert
(w_{\geq1})\\
&  =\left\Vert \left\langle \phi\right\rangle \varphi^{\prime}\right\Vert (w)
\end{align*}
for every $w\in A^{\ast}$, hence $\left\Vert E_{\varphi}\right\Vert
=\left\Vert \varphi\right\Vert $. \newline If $\varphi=\left\langle
\varphi_{1}?\right\rangle \varphi_{2}$ or $\varphi=\left\langle \rho_{1}%
\oplus\rho_{2}\right\rangle \varphi^{\prime}$ or $\varphi=\left\langle
\rho_{1}\cdot\rho_{2}\right\rangle \varphi^{\prime}$, then our claim holds
true by the induction hypothesis and the closure of the class $GRE(K,A)$ under
Hadamard product, sum, and Cauchy product, respectively. Finally, let
$\varphi=\left\langle \rho^{\oplus}\right\rangle \varphi^{\prime}$ and assume
that $\left\Vert \varphi \right \Vert$ is defined and there are generalized weighted rational expressions $E_{1},E_{2}$ such
that $\left\Vert E_{1}\right\Vert =\left\Vert \left\langle \rho\right\rangle
true\right\Vert $,  which is proper, and $\left\Vert E_{2}\right\Vert =\left\Vert \left\langle
\rho\right\rangle \varphi^{\prime}\right\Vert $. Then, we let $E_{\varphi
}=E_{1}^{+}\cdot E_{2}+E_{2}$ and for every $w\in A^{\ast}$ we get
\begin{align*}
\left\Vert E_{\varphi}\right\Vert (w)  &  =\left\Vert E_{1}^{+}\cdot E_{2}%
+E_{2}\right\Vert (w)\\
&  =\underset{w=uv,u\neq\varepsilon}{\sum}\left(  \left\Vert E_{1}%
^{+}\right\Vert (u)\cdot\left\Vert E_{2}\right\Vert (v)\right)  +\left\Vert
E_{2}\right\Vert (w)\\
&  =\underset{w=uv,u\neq\varepsilon}{\sum}\left(  \left\Vert E_{1}\right\Vert
^{+}(u)\cdot\left\Vert E_{2}\right\Vert (v)\right)  +\left\Vert E_{2}%
\right\Vert (w)\\
&  =\underset{w=uv,u\neq\varepsilon}{\sum} \left(\left\Vert \left\langle
\rho\right\rangle true\right\Vert ^{+}(u)\cdot\left\Vert \left\langle
\rho\right\rangle \varphi^{\prime}\right\Vert (v)\right)+\left\Vert \left\langle
\rho\right\rangle \varphi^{\prime}\right\Vert (w)\\
&  =\underset{w=uv,u\neq\varepsilon}{\sum}\underset{m \geq 1 }{\sum}\left(  \left\Vert \left\langle \rho\right\rangle
true\right\Vert ^{m}(u)\cdot\left\Vert \left\langle \rho\right\rangle
\varphi^{\prime}\right\Vert (v)\right)  +\left\Vert \left\langle
\rho\right\rangle \varphi^{\prime}\right\Vert (w)\\
&  =\underset{n \geq 2 }{\sum}\left\Vert
\left\langle \rho^{n}\right\rangle \varphi^{\prime}\right\Vert (w)+\left\Vert
\left\langle \rho\right\rangle \varphi^{\prime}\right\Vert (w)\\
&  =\underset{n\geq 1 }{\sum}\left\Vert
\left\langle \rho^{n}\right\rangle \varphi^{\prime}\right\Vert (w)\\
&  =\left\Vert \left\langle \rho^{\oplus}\right\rangle \varphi^{\prime
}\right\Vert (w),
\end{align*}
i.e., $\left\Vert E_{\varphi}\right\Vert =\left\Vert \left\langle \rho
^{\oplus}\right\rangle \varphi^{\prime}\right\Vert $ which concludes our
proof. \hfill $\square$ \medskip
\end{proof}

By Theorems \ref{RE_to_LDL} and \ref{LDL_to_RE} we get our first
main result.

\begin{theorem}
\label{LDL_eq_RE}Let $K$ be a semiring and $A$ an alphabet. Then a series
$s\in K\left\langle \left\langle A^{\ast}\right\rangle \right\rangle $ is
LDL-definable iff it is g-rational.
\end{theorem}

By Theorem \ref{LDL_eq_RE} and the discussion following Theorem \ref{Kleene}, we
immediately obtain the following consequence.

\begin{corollary}\label{LDL_eq_RE_com}
\label{LDL_to_REC} Let $K$ be a commutative semiring and $A$ an alphabet. A
series $s\in K\left\langle \left\langle A^{\ast}\right\rangle \right\rangle $
is LDL-definable iff it is recognizable.
\end{corollary}

The next proposition describes a doubly exponential translation of a weighted \emph{LDL} formula to an expressively equivalent weighted automaton. 

\begin{proposition}
\label{LDL_to_aut}
Let $K$ be a commutative semiring and $A$ an alphabet. For every weighted LDL formula $\varphi$ we can construct, in doubly exponential time, a weighted automaton $\mathcal{A}_{\varphi}$ such that $\left \Vert \mathcal{A}_{\varphi} \right \Vert = \left \Vert \varphi \right \Vert$.  
\end{proposition}
\begin{proof}
If $\varphi$ is an \emph{LDL} formula, then by \cite{Gi:Li,Gi:Sy} we get a deterministic finite automaton accepting the language of $\varphi$ which trivially can be considered as a weighted automaton with weights $0$ and $1$. Then, by applying structural induction on $\varphi$ we prove our claim by well-known constructions on weighted automata (cf. \cite{Dr:Au}). More precisely, for the closure under sum we take the disjoint union of two weighted automata and for Hadamard product the product automaton. For the closure under Cauchy product we firstly construct the corresponding normalized weighted automata with one initial and final state respectively, and then identify the final state of the first automaton with the initial state of the second automaton. Finally for the plus-iteration, we get firstly the normalized weighted automaton and extend it with a copy of it. Then, we identify the final state of the original automaton with the copy states corresponding to the initial state and final state. The new automaton has the same initial state and the merging one as its final state. Since the translation of an \emph{LDL} formula to a deterministic finite automaton is doubly exponential \cite{Gi:Li,Gi:Sy} and the aforementioned constructions on weighted automata are polynomial, we obtain a doubly exponential translation of weighted \emph{LDL} formulas to weighted automata. \hfill $\square$ \medskip
\end{proof}

The construction of the weighted automaton, as described in the above proposition, is not possible for any semiring, since, as is known \cite{Be:Ra}, there are non-commutative semirings $K$ and g-rational series $s \in K \left \langle \left \langle A^* \right \rangle \right \rangle$ which are not recognizable. On the other hand, it is well-known \cite{Dr:Han} that the equivalence of weighted automata is decidable whenever the weight structure is a computable field. More interestingly the complexity of checking the equivalence is cubic. Therefore, we get the third main result of our paper.

\begin{theorem}
Let $K$ be a computable field and $A$ an alphabet. Then, for every $\varphi, \varphi ' \in LDL(K,A)$ the equality $\left \Vert \varphi \right \Vert = \left \Vert \varphi ' \right \Vert$ is decidable in doubly exponential time.
\end{theorem}

\begin{corollary}
Let $K$ be a computable field, $A$ an alphabet, and $k \in K$. Then, for every $\varphi \in LDL(K,A)$ the equality $\left \Vert \varphi \right \Vert = \tilde{k}$ is decidable in doubly exponential time.
\end{corollary}

\begin{remark}
If $K$ is an idempotent commutative semiring, then for every weighted LDL formula $\varphi$ we can construct a weighted automaton $\mathcal{A}_{\varphi}$ such that $\left \Vert \mathcal{A}_{\varphi} \right \Vert = \left \Vert \varphi \right \Vert$ in exponential time. Indeed, if $\varphi$ is an LDL formula, then by \cite{Gi:Li,Gi:Sy} in exponential time we get a nondeterministic finite automaton accepting the language of $\varphi$, which, since $K$ is idempotent, can be considered as a weighted automaton with weights $0$ and $1$. Then proceed as before. In particular, if $K$ is a bounded distributive lattice, the equivalence of two weighted automata over $A$ and $K$ and hence of two weighted $LDL(K,A)$ formulas is again decidable \cite{Ra:Fu}.   
\end{remark}

\section{Weighted linear dynamic logic on infinite words}

In this section we interpret weighted \emph{LDL} formulas over infinite words. For this, we need our semiring to be equipped with infinite sums and
products. More precisely, we assume that the semiring $K$ is equipped, for
every index set $I$, with an infinitary sum operation $\sum_{I}:K^{I}\rightarrow
K$ such that for every family $(k_{i}\mid i\in I)$ of elements of $K$ and
$k\in K$ we have%
\begin{align*}
&  \displaystyle\sum_{i\in\emptyset}k_{i}=0,\hspace{1em}\sum_{i\in\{j\}}%
k_{i}=k_{j},\hspace{1em}\sum_{i\in\{j,l\}}k_{i}=k_{j}+k_{l}\mbox{ for
}j\neq l,\\
&  \displaystyle\sum_{j\in J}\Bigl(\sum_{i\in I_{j}}k_{i}\Bigr)=\sum_{i\in
I}k_{i}\mbox{, if $\bigcup_{j\in J}I_j=I$ and $I_j\cap
I_{j'}=\emptyset$ for $j\neq j'$,}\\
&  \displaystyle\sum_{i\in I}(k\cdot k_{i})=k\cdot\Bigl(\sum_{i\in I}%
k_{i}\Bigr),\hspace{1em}\sum_{i\in I}(k_{i}\cdot k)=\Bigl(\sum_{i\in I}%
k_{i}\Bigr)\cdot k.
\end{align*}
Then the semiring $K$ together with the operations $\sum_{I}$\ is called
\emph{complete} \ \cite{Ei:Au,Ku:Se}.

A complete semiring is said to be \emph{totally complete\/ }\cite{Es:On}%
\emph{, }if it is endowed with a countably infinite product operation
satisfying for every sequence $(k_{i}\mid i\geq0)$ of elements of $K$ the
subsequent conditions:%
$$
\prod_{i\geq0}1=1,\hspace{1em}\prod_{i\geq0}k_{i}=\prod_{i\geq0}k_{i}^{\prime
},\hspace{1em}
k_{0}\cdot\prod_{i\geq0}k_{i+1}=\prod_{i\geq0}k_{i},\hspace{1em}\prod_{j\geq
1}\sum_{i\in I_{j}}k_{i}=\sum_{(i_{1},i_{2},\ldots)\in I_{1}\times I_{2}%
\times\ldots}\prod_{j\geq1}k_{i_{j}},
$$
where in the second equation $k_{0}^{\prime}=k_{0}\cdot\ldots\cdot k_{n_{1}%
},k_{1}^{\prime}=k_{n_{1}+1}\cdot\ldots\cdot k_{n_{2}},\ldots$ for any
increasing sequence $0<n_{1}<n_{2}<\ldots,$ and in the last equation
$I_{1},I_{2},\dots$ are arbitrary index sets.

Furthermore, we will call a totally complete semiring $K$ \emph{totally
commutative complete }if it satisfies the equation:%
\[
\prod_{i\geq0}\left(  k_{i}\cdot k_{i}^{\prime}\right)  =\left(
\underset{i\geq0}{\prod}k_{i}\right)  \cdot\left(  \prod_{i\geq0}k_{i}%
^{\prime}\right)  .
\]
Obviously a totally commutative complete semiring is commutative. We refer the reader to \cite{Dr:Han,Ei:Au,Ku:Se} for examples of complete semirings. Throughout this section we assume $K$ to be a totally complete semiring. An \emph{infinitary series} (or simply \emph{series}) over $A^{\omega}$ and $K$ is a mapping $s:A^{\omega} \rightarrow K$. We denote by $K\left\langle \left\langle A^{\omega} \right\rangle
\right\rangle $ the class of all series over $A^{\omega}$ and $K$.  The sum, the products with scalars, and the Hadamard product of series in $ K\left\langle \left\langle
A^{\omega}\right\rangle \right\rangle $ are defined elementwise as for series on finite words. The structure $\left(  K\left\langle \left\langle A^{\omega}\right\rangle \right\rangle
,+,\odot,\widetilde{0},\widetilde{1}\right)  $ of infinitary series over $A$ and $K$ is a totally complete semiring.  Next let $s\in K\left\langle \left\langle
A^{\ast}\right\rangle \right\rangle $ and $r\in K\left\langle \left\langle
A^{\omega}\right\rangle \right\rangle $). The \emph{Cauchy product} 
$s\cdot r\in K\left\langle \left\langle A^{\omega}\right\rangle \right\rangle
$) is determined by $(s\cdot
r)(w)=\sum\nolimits_{w=uv,u\in A^{\ast}}s(u)r(v)$ for every $w\in A^{\omega}$. Finally, the
$\omega$\emph{-iteration} $s^{\omega}\in K\left\langle \left\langle A^{\omega
}\right\rangle \right\rangle $ \emph{of a proper series }$s\in K\left\langle
\left\langle A^{\ast}\right\rangle \right\rangle $ is defined by $s^{\omega
}(w)=\sum\nolimits_{w=w_{0}w_{1}\ldots}\prod\nolimits_{i\geq0}s(w_{i}%
)$.

Next, we recall
\emph{weighted }$\omega$\emph{-rational expressions over }$A$ \emph{and} $K$
which are defined by the grammar $E::=E+E\mid F\cdot E\mid F^{\omega}$ 
where $F$ is any weighted rational expression. We denote by $\omega$-$RE(K,A)$
the class of all such weighted $\omega$-rational expressions over $A$ and $K$. 
Similarly we define the class of \emph{generalized weighted }$\omega
$-\emph{rational expressions over }$A$ \emph{and} $K$ which is given by the
grammar $E::=E+E\mid F\cdot E\mid F^{\omega}\mid E\odot E $, where $F$ is any generalized weighted rational expression. We shall denote by
$\omega$-$GRE(K,A)$ the class of generalized weighted $\omega$-rational
expressions over $A$ and $K$. The \emph{semantics} of a (generalized) weighted
$\omega$-rational expression $E$ is a series $\left\Vert E\right\Vert \in
K\left\langle \left\langle A^{\omega}\right\rangle \right\rangle $ which is
defined inductively by $\left\Vert E+E^{\prime}\right\Vert =\left\Vert E\right\Vert
+\left\Vert E^{\prime}\right\Vert , \ \ \left\Vert F\cdot E\right\Vert =\left\Vert F\right\Vert
\cdot\left\Vert E\right\Vert ,  \ \ \left\Vert F^{\omega}\right\Vert =\left\Vert F\right\Vert ^{\omega
}$ (if $\left \Vert F \right \Vert$ is proper; otherwise undefined), \ \ $\left\Vert E\odot E^{\prime}\right\Vert =\left\Vert E\right\Vert
\odot\left\Vert E^{\prime}\right\Vert$. A series $s\in K\left\langle \left\langle A^{\omega}\right\rangle
\right\rangle $ is called $\omega$-\emph{rational }(resp. \emph{g-}$\omega
$\emph{-rational}) if there is a weighted (resp. generalized weighted) $\omega
$-rational expression $E$ such that $s=\left\Vert E\right\Vert $. The subsequent result states the coincidence of $\omega$-rational and $\omega$-recognizable series, i.e., infinitary series accepted by weighted automata over infinite words. For the theory on weighted automata over infinite words we refer the reader to \cite{Es:Ha,Dr:Wh}. 
\begin{theorem}
\label{Kleene-omega}\cite{Es:Ha} 
Let $K$ be a totally complete semiring
and $A$\ an alphabet. Then a series $s\in K\left\langle \left\langle
A^{\omega}\right\rangle \right\rangle $ is $\omega$-rational iff it is $\omega$-recognizable.
\end{theorem}

It is well-known (cf. \cite{Dr:Wh}) that if the semiring $K$ is totally commutative complete, then the class of 
$\omega$-recognizable series over $A$ and $K$ is closed under Hadamard product.
Consequently, if $K$ is totally commutative complete,
then a series $s\in K\left\langle \left\langle
A^{\omega}\right\rangle \right\rangle $ is g-$\omega$-rational iff it is
$\omega$-recognizable.

We shall need to extend the syntax of \emph{LDL} formulas and weighted
\emph{LDL }formulas as follows.

\begin{definition}
\label{LDL_inf_synt}\cite{Va:Th} The syntax of formulas $\xi$\ of the LDL over $A$,
interpreted over infinite words, is given by the grammar
\begin{align*}
\xi &  ::=true\mid p_{a}\mid\lnot\xi\mid\xi\wedge\xi\mid\left\langle
\eta\right\rangle \xi\\
\eta &  ::=\phi\mid\xi?\mid\eta+\eta\mid\theta;\eta\mid\theta^{\omega}%
\end{align*}
where $p_{a}\in P$, $\phi$ denotes a propositional formula over the atomic
propositions in $P$, and $\theta$ denotes an expression as in Definition \ref{LDL_fin_synt}.
\end{definition}

For every \emph{LDL} formula $\xi$ and $w\in A^{\omega}$ we define the satisfaction relation
$w\models\xi$, inductively on the structure of $\xi$, as follows:
\begin{itemize}
\item[-] $w\models true,$

\item[-] $w\models p_{a}$ \ iff $\ w(0)=a,$

\item[-] $w\models\lnot\xi$ \ iff $\ w\not \models \xi,$

\item[-] $w\models\xi_{1}\wedge\xi_{2}$ \ iff $\ w\models\xi_{1}$ and
$w\models\xi_{2},$

\item[-] $w\models\left\langle \phi\right\rangle \xi$ \ iff $\ w\models\phi$
and $w_{\geq1}\models\xi,$

\item[-] $w\models\left\langle \xi_{1}?\right\rangle \xi_{2}$ \ iff
$\ w\models\xi_{1}$ and $w\models\xi_{2},$

\item[-] $w\models\left\langle \eta_{1}+\eta_{2}\right\rangle \xi$ \ iff
$\ w\models\left\langle \eta_{1}\right\rangle \xi$ or $w\models\left\langle
\eta_{2}\right\rangle \xi,$

\item[-] $w\models\left\langle \theta;\eta\right\rangle \xi$ \ iff $\ w=uv$
with $u\in A^{\ast}$, $u\models\left\langle \theta\right\rangle true$, and
$v\models\left\langle \eta\right\rangle \xi,$

\item[-] $w\models\left\langle \theta^{\omega}\right\rangle \xi$ \ iff
\ $\xi=true$, $w=w_{0}w_{1}\ldots$, and $w_{i}\models\left\langle
\theta\right\rangle true$ for every $i\geq0.$
\end{itemize}

For an \emph{LDL} formula $\xi$, we let $L_{\omega}(\xi)=\{w \in A^{\omega} \mid w \models \xi  \}$, the infinitary language defined by $\xi$. An infinitary language $L\in A^{\omega
}$ is called \emph{LDL}-$\omega$-definable if there is an \emph{LDL} formula
$\xi$ such that $L=L_{\omega}(\xi)$. The coincidence of $\omega$-rational and \emph{LDL}-$\omega$-definable languages was stated in \cite{Va:Th}.

\begin{theorem}
\label{LDL_RE_omega}\cite{Va:Th} A language $L\subseteq A^{\omega}$ is
LDL-$\omega$-definable iff $L$ is $\omega$-rational.
\end{theorem}

Next we introduce the syntax of the weighted \emph{LDL} formulas
interpreted over infinite words.

\begin{definition}
\label{wLDL_inf_synt}The syntax of formulas $\zeta$\ of the \emph{weighted}
LDL \emph{over} $A$ \emph{and} $K$, interpreted over infinite words, is given
by the grammar
\begin{align*}
\zeta &  ::=k\mid\xi\mid\zeta\oplus\zeta\mid\zeta\otimes\zeta\mid\left\langle
\pi\right\rangle \zeta\\
\pi &  ::=\phi\mid\zeta?\mid\pi\oplus\pi\mid\rho\cdot\pi\mid\rho^{\varpi}%
\end{align*}
where $k\in K$, $p_{a}\in P$, $\phi$ denotes a propositional formula over
the atomic propositions in $P$, $\xi$ denotes an LDL formula as in Definition \ref{LDL_inf_synt}, and $\rho$ an expression as in Definition \ref{wLDL_fin_synt}. 
\end{definition}

We denote by $LDL_{\omega}(K,A)$ the set of all weighted \emph{LDL} formulas
$\zeta$ over $A$ and $K$. We represent the semantics $\left\Vert
\zeta\right\Vert _{\omega}$ of formulas $\zeta\in LDL_{\omega}(K,A)$ as series
in $K\left\langle \left\langle A^{\omega}\right\rangle \right\rangle $. For
the semantics of \emph{LDL} formulas $\xi$\ interpreted over infinite words,
we use the satisfaction relation $\models$ as defined above. 

\begin{definition}
\label{def-sem_omega}Let $\zeta\in LDL_{\omega}(K,A)$. The \emph{semantics} of
$\zeta$ is a series $\left\Vert \zeta\right\Vert _{\omega}\in K\left\langle
\left\langle A^{\omega}\right\rangle \right\rangle $. For every $w\in
A^{\omega}$ the value $\left\Vert \zeta\right\Vert _{\omega}(w)$ is defined
inductively as follows:

$\begin{array}[c]{ll}
\left\Vert k\right\Vert_{\omega} (w)=k, &  \ \  \  \ \ \  \ \ \ \ \  \ \left\Vert \zeta_{1}\oplus\zeta_{2}\right\Vert_{\omega} (w)=\left\Vert
\zeta_{1}\right\Vert_{\omega} (w)+\left\Vert \zeta_{2}\right\Vert_{\omega} (w),  \\
\left\Vert \xi\right\Vert_{\omega} (w)=\left\{
\begin{array}
[c]{rl}%
1 & \text{if }w\models\xi\\
0 & \text{otherwise}%
\end{array}
\right.  ,       &    \ \  \  \  \ \ \ \ \ \ \  \    \left\Vert \zeta_{1}\otimes\zeta_{2}\right\Vert_{\omega} (w)=\left\Vert
\zeta_{1}\right\Vert_{\omega} (w)\cdot\left\Vert \zeta_{2}\right\Vert_{\omega} (w),  \\
\left\Vert \left\langle \phi\right\rangle \zeta\right\Vert_{\omega}
(w)=\left\Vert \phi\right\Vert_{\omega} (w)\cdot\left\Vert \zeta\right\Vert_{\omega} (w_{\geq
1}),    &  \ \  \ \  \ \ \ \ \  \ \  \  \left\Vert \left\langle \zeta_{1}?\right\rangle \zeta_{2}\right\Vert_{\omega}
(w)=\left\Vert \zeta_{1}\right\Vert_{\omega} (w)\cdot\left\Vert \zeta_{2}\right\Vert_{\omega}
(w),     
\end{array}$

$\begin{array}[c]{ll}
\left\Vert \left\langle \pi_{1}\oplus\pi_{2}\right\rangle
\zeta\right\Vert_{\omega} (w)=\left\Vert \left\langle \pi_{1}\right\rangle
\zeta\right\Vert_{\omega} (w)+\left\Vert \left\langle \pi_{2}\right\rangle
\zeta\right\Vert_{\omega} (w),  \\
\left\Vert \left\langle \rho\cdot\pi\right\rangle \zeta\right\Vert_{\omega}
(w)=\underset{w=uv,u\in A^{\ast}}{\sum}\left(  \left\Vert \left\langle
\rho\right\rangle true\right\Vert (u)\cdot\left\Vert \left\langle
\pi\right\rangle \zeta\right\Vert_{\omega} (v)\right)  ,\\
\left\Vert \left\langle \rho^{\varpi}\right\rangle \zeta\right\Vert_{\omega}
(w)=\left\{
\begin{tabular}
[c]{ll}%
$\underset{w=w_{0}w_{1}\ldots}{\sum}\underset{i\geq0}{\prod}\left\Vert
\left\langle \rho\right\rangle true\right\Vert (w_{i})$ & if $\zeta=true$\\
$0$ & otherwise
\end{tabular}
\ \ ,\right.
\end{array}$

\noindent where for the definition of $\left\Vert \left\langle \rho^{\varpi}\right\rangle \zeta\right\Vert_{\omega}
(w)$ we assume that $\left\Vert \left\langle \rho \right\rangle true \right\Vert$ is proper.
\end{definition}

A series $s\in K\left\langle \left\langle A^{\omega}\right\rangle
\right\rangle $ is called \emph{LDL}-$\omega$-\emph{definable }if there is a
formula $\zeta\in LDL_{\omega}(K,A)$ such that $s=\left\Vert \zeta\right\Vert
_{\omega}$.
For $K=\mathbb{B}$ and any $L \subseteq A^{\omega}$, clearly $L$ is \emph{LDL}-$\omega$-definable iff $1_L \in \mathbb{B} \left\langle \left\langle A^{\omega}\right\rangle \right\rangle$ is \emph{LDL}-$\omega$-definable, and therefore our weighted \emph{LDL} generalizes \emph{LDL} over infinite words.

\begin{example}\label{ex3}
Let $(\mathbb{N} \cup \{\infty\} ,+, \cdot, 0, 1)$ be the totally complete semiring of extended natural numbers, $A = \{a,b\}$, and $k \in \mathbb{N} \setminus \{0\} $. We consider the LDL formula $\psi_1= \left\langle \left (\left ( \left\langle p_b? \right \rangle Last \right ) ? \right )^{\oplus} \right \rangle true \vee (\lnot p_a \wedge \lnot p_b)$,  
the weighted LDL formula $\psi_2=\left \langle (k \otimes p_a)? \right \rangle Last$, and we let 
$$ \zeta = \left\langle \left( ( \psi_1? \cdot \psi_2? \cdot \psi_1? \cdot \psi_2?)^{\oplus}    \oplus (\lnot p_a \wedge \lnot p_b)? \right ) \cdot \left( \left( \left \langle p_b? \right \rangle Last \right)? \right)^{\varpi}  \right \rangle true.$$
By a standard computation we can show that for every $w \in A^{\omega}$ we get $\left \Vert \zeta \right \Vert_{\omega} (w)= k^{|w|_a}$ whenever $|w|_a < \infty$ and it is even, and $\left \Vert \zeta \right \Vert_{\omega} (w)=0$ otherwise. Furthermore, since the infinitary language $L= \{ w \in A^{\omega} \mid w \text{ contains an even number of } a's \}$ is not $\omega$-star-free (cf. \cite{Mu:No}), with a similar argument as in Example \ref{ex1}, we can show that the series $\left \Vert \zeta \right \Vert_{\omega}$ is not $\omega$-definable by any weighted FO logic sentence (resp. LTL formula) (cf. Section 5 and \cite{Ma:Ph,Ma:Se}) over the extended naturals.  
\end{example}

The next theorem states that every generalized weighted $\omega$-rational
expression can be translated to a weighted \emph{LDL} formula in linear time. The proof is done by induction on the structure of generalized weighted $\omega$-rational
expressions, as in the proof of Theorem \ref{RE_to_LDL}.

\begin{theorem}
\label{RE_to_LDL_omega}For every generalized weighted $\omega$-rational
expression $E\in\omega$-$GRE(K,A)$ we can construct, in linear time, a
weighted LDL formula $\zeta_{E}\in LDL_{\omega}(K,A)$ with $\left\Vert
\zeta_{E}\right\Vert _{\omega}=\left\Vert E\right\Vert $.
\end{theorem}\smallskip

In the sequel, we show that also the converse result holds. For this, we need the subsequent lemma.

\begin{lemma}
\label{reg_expr_0_1_omega}Let $E$ be an $\omega$-rational expression over
$A$\ and $L(E)$ the language defined by $E$. Then, there is an $E^{\prime}%
\in\omega$-$RE(K,A)$ such that $\left\Vert E^{\prime}\right\Vert (w)=1$ if
$w\in L(E)$ and $\left\Vert E^{\prime}\right\Vert (w)=0$ otherwise, for every
$w\in A^{\omega}$.
\end{lemma}

\begin{theorem}
\label{LDL_to_RE_omega}For every weighted LDL formula $\zeta\in LDL_{\omega
}(K,A)$ we can construct a generalized weighted $\omega$-rational expression
$E_{\zeta}\in\omega$-$GRE(K,A)$ such that $\left\Vert E_{\zeta}\right\Vert
=\left\Vert \zeta\right\Vert _{\omega}$.
\end{theorem} \smallskip
\begin{proof}
[Sketch] By induction on the structure of $LDL_{\omega}(K,A)$
formulas $\zeta$, using similar arguments as the ones in the proof of Theorem \ref{LDL_to_RE}. More precisely, if $\zeta=\xi$ is an \emph{LDL} formula, then we use Lemma \ref{reg_expr_0_1_omega}. For the induction steps, we use the closure of generalized
weighted $\omega$-rational expressions under sum, Hadamard and Cauchy products, and $\omega$-iteration. \hfill $\square$ \medskip
\end{proof}

By Theorems \ref{RE_to_LDL_omega} and \ref{LDL_to_RE_omega} we get the fourth 
main result of our paper.

\begin{theorem}
\label{LDL_eq_RE_omega}Let $K$ be a totally complete semiring and $A$ an
alphabet. Then a series $s\in K\left\langle \left\langle A^{\omega
}\right\rangle \right\rangle $ is LDL-$\omega$-definable iff it is g-$\omega$-rational.
\end{theorem}

By Theorem  \ref{LDL_eq_RE_omega} and the discussion following Theorem \ref{Kleene-omega} we
get the subsequent corollary.

\begin{corollary}
\label{LDL_to_MSO_omega}Let $K$ be a totally commutative complete semiring and
$A$ an alphabet. A series $s\in K\left\langle \left\langle A^{\omega
}\right\rangle \right\rangle $ is LDL-$\omega$-definable iff it is $\omega$-recognizable.
\end{corollary}

\begin{proposition}
\label{LDL_to_omega_aut}
Let $K$ be an idempotent totally commutative complete semiring and $A$ an alphabet. For every weighted LDL formula $\zeta$ we can construct, in exponential time, a weighted B\"{u}chi automaton $\mathcal{A}_{\zeta}$ such that $\left \Vert \mathcal{A}_{\zeta} \right \Vert = \left \Vert \zeta \right \Vert_{\omega} $.  
\end{proposition}
\begin{proof}
If $\zeta$ is an \emph{LDL} formula, then it is an \emph{PLDL} (\emph{parametric linear dynamic logic}) formula and, by \cite{Fa:Pa} we get in exponential time a nondeterministic B\"{u}chi automaton accepting the language of $\zeta$. This automaton can be considered as a weighted B\"{u}chi automaton with weights $0$ and $1$. Then, by applying structural induction on $\zeta$ we prove our claim by standard constructions on weighted B\"{u}chi automata. More precisely, for the closure under sum we take the disjoint union of two weighted B\"{u}chi automata. For Hadamard product we use the  well-known product construction for B\"{u}chi automata, showing the closure of the class of $\omega$-recognizable languages under intersection \cite{Th:Au}, reasonably translated to weighted setup. For the closure under Cauchy product we construct the corresponding normalized weighted automaton and initial weight normalized  weighted B\"{u}chi automaton, and then identify the final state of the first automaton with the initial state of the second automaton. 
Finally, for the $\omega$-iteration, we again get the normalized weighted automaton and identify its initial and final state. All the aforementioned constructions are polynomial, and our proof is completed. \hfill $\square$ \medskip
\end{proof}
  
In particular, if $K$ is a bounded distributive lattice, the equivalence of two weighted automata over $A$ and $K$ on infinite words and hence of two $LDL(K,A)$ formulas is again decidable \cite{Dr:Ra}.

\section{Comparison of weighted \emph{LDL} to other weighted logics}

In this last section we state the relation of our weighted \emph{LDL} to weighted monadic second-order logic (weighted \emph{MSO} logic for short), weighted linear temporal logic (weighted \emph{LTL} for short) and weighted $\mu$-calculus. The relation of \emph{LDL}-definable series (resp. infinitary series) to weighted \emph{MSO} logic definable series (resp. infinitary series) is immediately derived by \cite{Dr:We,Dr:Wh} and Corollary \ref{LDL_to_REC} (resp. by \cite{Dr:Wh} and Corollary \ref{LDL_to_MSO_omega}). We get the following consequences. 

\begin{corollary}
\label{LDL_to_MSO}Let $K$ be a commutative semiring and $A$ an alphabet. A
series $s\in K\left\langle \left\langle A^{\ast}\right\rangle \right\rangle $
is LDL-definable iff it is definable by a restricted weighted MSO logic
sentence over $A$ and $K$.
\end{corollary}

\begin{corollary}
\label{last_col}Let $K$ be a totally commutative complete semiring and
$A$ an alphabet. A series $s\in K\left\langle \left\langle A^{\omega
}\right\rangle \right\rangle $ is LDL-$\omega$-definable iff it is definable
by a restricted weighted MSO logic sentence over $A$ and $K$ interpreted over
infinite words.
\end{corollary}

Weighted \emph{LTL} has been investigated over De Morgan algebras \cite{Ku:La}, arbitrary bounded lattices \cite{Dr:Mu}, idempotent zero-divisor free totally commutative complete semirings \cite{Ma:Ph,Ma:Se}, with averaging modalities \cite{Bo:Av}, with discounting over the interval $[0,1]$ \cite{Al:Fo,Al:Di}, and with discounting over the max-plus semiring \cite{Ma:Ph,Ma:We}. Recently, a type of weighted \emph{LTL} has been applied to robotics \cite{La:Th}.
We need to recall
first the classical \emph{LTL} (cf. \cite{Ba:Pr}). For every
letter $a\in A$ we consider an atomic proposition $p_{a}$ and we let
$P=\{p_{a}\mid a\in A\}$.  The syntax of \emph{LTL} formulas over $A$\ is given
by the grammar $\phi::=true\mid p_{a}\mid\lnot\phi\mid\phi\vee\phi\mid\bigcirc\phi\mid\phi U\phi$ 
where $p_{a}\in P$. Let $\phi$ be an \emph{LTL }formula over $A$. For every
$w=a_{0}\ldots a_{n-1}\in A^{\ast}$ and $0\leq i\leq n-1$ (resp. $w=a_{0}a_{1}\ldots\in A^{\omega}$ and $i\geq0$) the satisfaction
relation $w,i\models\phi$ is defined as usual (cf. for instance \cite{Ba:Pr,Di:Fi}) by induction on the structure of $\phi$.

The syntax of formulas $\varphi$ of the \emph{weighted LTL over} $A$ \emph{and}
$K$ is given by the grammar
\begin{align*}
\varphi &  ::=k\mid\phi\mid\varphi\oplus\varphi\mid\varphi\otimes\varphi
\mid\pmb{\bigcirc}\varphi\mid\varphi\mathcal{U}\varphi\mid\boxtimes\varphi
\end{align*}
where $k\in K$, $p_{a}\in P$, and $\phi$ is an \emph{LTL} formula over $A$.

We denote by $LTL(K,A)$ the class of all weighted \emph{LTL} formulas
$\varphi$ over $A$ and $K$. Firstly, we represent the semantics $\left\Vert
\varphi\right\Vert $ of formulas $\varphi\in LTL(K,A)$ as series in
$K\left\langle \left\langle A^{\ast}\right\rangle \right\rangle $. For the
semantics of \emph{LTL} formulas $\phi$ we use the satisfaction relation as
defined above. 

\begin{definition}
\label{def-ltl}Let $\varphi\in LTL(K,A)$. The \emph{semantics} of $\varphi$ is
a series $\left\Vert \varphi\right\Vert \in K\left\langle \left\langle
A^{\ast}\right\rangle \right\rangle $. For every $w\in A^{\ast}$, with $|w|=n$ ($n \geq 0$), the value
$\left\Vert \varphi\right\Vert (w)$ is defined inductively as follows:

$\begin{array}[c]{ll}
\left\Vert k\right\Vert (w)=k, & \ \ \ \ \ \ \ \ \ \left\Vert \varphi\oplus\psi\right\Vert (w)=\left\Vert
\varphi\right\Vert (w)+\left\Vert \psi\right\Vert (w),\\
\left\Vert \phi\right\Vert (w)=\left\{
\begin{array}
[c]{ll}%
1 & \text{if }w\models\phi\\
0 & \text{otherwise}%
\end{array}
\right.  ,     & \ \ \ \ \ \ \ \ \ \left\Vert \varphi\otimes\psi\right\Vert (w)=\left\Vert
\varphi\right\Vert (w)\cdot\left\Vert \psi\right\Vert (w), \\
\left\Vert \pmb{\bigcirc}\varphi\right\Vert (w)=\left\Vert
\varphi\right\Vert (w_{\geq1}), & \ \ \ \ \ \ \ \ \  \left\Vert \boxtimes\varphi\right\Vert (w)=\underset{0\leq i\leq
n-1}{%
{\displaystyle\prod}
}\left\Vert \varphi\right\Vert (w_{\geq i}),
\end{array}$ 

$\begin{array}[c]{ll}
\left\Vert \varphi\mathcal{U}\psi\right\Vert (w)=\underset{0\leq
i\leq n-1}{%
{\displaystyle\sum}
}\left(  \left(  \underset{0\leq j<i}{%
{\displaystyle\prod}
}\left\Vert \varphi\right\Vert (w_{\geq j})\right)  \cdot\left\Vert
\psi\right\Vert (w_{\geq i})\right).
\end{array}$
\end{definition}

A series $s\in K\left\langle \left\langle A^{\ast
}\right\rangle \right\rangle $ is called \emph{LTL-definable} if there is
a formula $\varphi\in LTL\left(  K,A\right)  $ such that $s=\left\Vert
\varphi\right\Vert $.

\begin{example}\label{ex4}
We consider the semiring $(\mathbb{N},+,\cdot,0,1)$ of natural numbers and the LTL formulas $\varphi = \boxtimes 2$ and $\psi=\boxtimes \varphi$. Then, we can easily see that for every $w \in A^*$, we get $\|\varphi \|(w)=2^{|w|}$ and $\|\psi \|(w)=2^{2^{|w|}}$. It is well known (cf. Ex. 3.4 in \cite{Dr:We}) that the series $\|\psi\|$ is not recognizable, and hence by Corollary \ref{LDL_eq_RE_com} not LDL-definable. 
\end{example}

By Examples \ref{ex1} and \ref{ex4} we immediately obtain the following proposition.

\begin{proposition}\label{LTL_inc_LDL}
The classes of LDL-definable and LTL-definable series over the semiring of natural numbers and an alphabet $A$ are incomparable.
\end{proposition}

Next, we represent the semantics of formulas in $LTL(K,A)$ as infinitary 
series in $K\left\langle \left\langle A^{\omega}\right\rangle \right\rangle $.

\begin{definition}
\label{def-omega-ltl}Let $K$ be a totally complete semiring and $\varphi\in
LTL(K,A)$. The \emph{semantics} \emph{of }$\varphi$ \emph{over infinite words}
is an infinitary series $\left\Vert \varphi\right\Vert _{\omega}\in
K\left\langle \left\langle A^{\omega}\right\rangle \right\rangle $. For every
$w\in A^{\omega}$ the value $\left\Vert \varphi\right\Vert _{\omega}(w)$ is
defined inductively as in the case of finite words except for the operators
$\mathcal{U}$ and $\boxtimes$:

$\begin{array}[c]{ll}
\left\Vert \varphi\mathcal{U}\psi\right\Vert _{\omega}%
(w)=\underset{i\geq0}{%
{\displaystyle\sum}
}\left(  \left(  \underset{0\leq j<i}{%
{\displaystyle\prod}
}\left\Vert \varphi\right\Vert_{\omega} (w_{\geq j})\right)  \cdot\left\Vert
\psi\right\Vert_{\omega} (w_{\geq i})\right)  , \\
\left\Vert \boxtimes\varphi\right\Vert _{\omega}(w)=\underset
{i\geq0}{%
{\displaystyle\prod}
}\left\Vert \varphi\right\Vert_{\omega} (w_{\geq i}).
\end{array}$
\end{definition}

A series  $s\in K\left\langle \left\langle
A^{\omega}\right\rangle \right\rangle $ is called \emph{LTL}-$\omega$\emph{-definable} if there is
a formula $\varphi\in LTL\left(  K,A\right)  $ such that  $s=\left\Vert \varphi\right\Vert _{\omega}$.
In view of Proposition \ref{LTL_inc_LDL}, we define a fragment of our weighted \emph{LTL}, and show that the class of series (resp. infinitary series) defined by \emph{LTL} formulas in this fragment is in the class of \emph{LDL}-definable (resp. \emph{LDL}-$\omega$-definable) ones. More precisely, an
\emph{LTL-step formula} is an $LTL(K,A)$ formula of the form $\oplus_{1\leq
i\leq n}\left(  k_{i}\otimes\varphi_{i}\right)  $ where $k_{i}\in K$ and
$\varphi_{i}$ is an \emph{LTL }formula for every $1\leq i\leq n$. Then, we call a 
formula $\varphi\in LTL\left(  K,A\right)  $ \emph{restricted} if
whenever it contains a subformula of the form $\boxtimes\psi$ or
$\psi\mathcal{U}\xi$,\ then $\psi$ is an \emph{LTL}-step formula. We shall denote by $rLTL(K,A)$ the set of all restricted $LTL(K,A)$ formulas. 
A series $s\in K\left\langle \left\langle A^{\ast
}\right\rangle \right\rangle $ (resp. $s\in K\left\langle \left\langle
A^{\omega}\right\rangle \right\rangle $) is called \emph{rLTL}%
-\emph{definable} (resp. \emph{rLTL}-$\omega$\emph{-definable} ) if there is
a formula $\varphi\in rLTL\left(  K,A\right)  $ such that $s=\left\Vert
\varphi\right\Vert $ (resp. $s=\left\Vert \varphi\right\Vert _{\omega}$). 
By an inductive construction, we can show that every \emph{rLTL}-definable (resp. \emph{rLTL}-$\omega$-definable) series is also definable (resp. $\omega$-definable) by a restricted weighted \emph{FO} logic sentence in the sense of \cite{Dr:We}. Therefore, by Corollaries \ref{LDL_to_MSO} and \ref{last_col}, we get respectively, the subsequent results.
  
\begin{theorem}\label{ULTL_to_LDL}
Let $K$ be a commutative semiring and $A$ an alphabet. If a series $s\in K\left\langle \left\langle A^{\ast}\right\rangle
\right\rangle $ is rLTL-definable, then it is LDL-definable.
\end{theorem}

\begin{theorem}\label{ULTL_to_LDL_omega}
Let $K$ be a totally commutative complete semiring and $A$ an
alphabet. If a series $s\in K\left\langle \left\langle A^{\omega}\right\rangle
\right\rangle $ is rLTL-$\omega$-definable, then it is LDL-$\omega$-definable.
\end{theorem}

A weighted $\mu$-calculus over a particular class of semirings was investigated in \cite{Me:We} (cf. also \cite{La:Al}). More precisely, the author showed that the class of rational (resp. $\omega$-rational) series over dc-semirings with the Arden fixed point property (resp. with infinite products and the Arden fixed point property) coincides with the class of series (resp. infinitary series) definable by the weighted conjunction-free $\mu$-calculus. Therefore, by Corollaries \ref{LDL_eq_RE_com}, \ref{LDL_to_MSO_omega} and Theorem 4.5 in \cite{Me:We}, we immediately obtain the following theorem.

\begin{theorem}
Let $K$ be a commutative (resp. totally commutative complete) dc-semiring with the Arden fixed point property and $A$ an alphabet. Then a series $s\in K\left\langle \left\langle A^{*}\right\rangle \right\rangle $ (resp. $s\in K\left\langle \left\langle A^{\omega}\right\rangle \right\rangle $) is LDL-definable (resp. LDL-$\omega$-definable)  iff it is definable by a sentence of the weighted conjunction-free $\mu$-calculus over $A$ and $K$. 
\end{theorem}

\section{Conclusion}

We introduced a weighted linear dynamic logic for finite (resp. infinite) words over
arbitrary (resp. totally complete) semirings and proved the expressive
equivalence of formulas of this logic with generalized weighted rational (resp.
$\omega$-rational) expressions. In our proofs we used structural induction for both
directions. We proved also that the translation of any weighted \emph{LDL} formula to a weighted automaton can be done as well, by structural induction, using the corresponding translation of \cite{Gi:Li,Gi:Sy} and well-known constructions on weighted automata. More interestingly, for the applications, the time complexity of the translation does not increase in the weighted setup. We recalled the weighted \emph{LTL} and showed that the class of series defined by weighted \emph{LTL} and weighted \emph{LDL} formulas are, in general, incomparable, in contrast to the well known relation for classical logics. We defined a fragment of weighted \emph{LTL}, which is larger than the one in recent works \cite{Ma:Ph,Ma:Se}, and showed that \emph{LTL}-definable (resp. \emph{LTL}-$\omega$-definable) series in this fragment are also \emph{LDL}-definable (resp. \emph{LDL}-$\omega$-definable). Recent applications require weighted automata (resp. weighted automata with input infinite words) over more general structures than semirings, for instance incorporating average or discounted computations of weights \cite{Ch:Qu,Dr:Re,Dr:Av}. Therefore, it should be very interesting, especially for applications, to explore the expressive power of a weighted \emph{LDL} over more  general weight structures.

\nocite{*}
\bibliographystyle{eptcs}
\bibliography{paper-bibliography-final}





\end{document}